\documentclass[12pt]{iopart}
\usepackage{amsbsy,amsthm,graphicx}

\begin{document}

\newtheorem{theo}{Theorem} \newtheorem{lemma}{Lemma}

\newcommand{\mss}[1]{\mbox{\scriptsize #1}}

\title{Coherent feedback that beats all measurement-based feedback protocols}
\author{Kurt Jacobs$^{1,2,3}$, Xiaoting Wang$^{1,2,4}$, and Howard M. Wiseman$^{5}$}
\address{$^1$ Advanced Science Institute, RIKEN, Wako-shi 351-0198, Japan}
\address{$^2$ Department of Physics, University of Massachusetts at Boston, 100 Morrissey
Blvd, Boston, MA 02125, USA} 
\address{$^3$ Hearne Institute for Theoretical Physics, Louisiana State University, Baton Rouge, LA 70803, USA}
\address{$^4$ Research Laboratory of Electronics, Massachusetts Institute of Technology, Cambridge, MA 02139, USA}
\address{$^5$ Centre for Quantum Computation and Communication Technology (Australian Research Council), Centre for Quantum Dynamics, Griffith University, Brisbane, Queensland 4111, Australia}

\ead{kurt.jacobs@umb.edu} 

\begin{abstract}
We show that when the speed of control is bounded, there is a widely applicable minimal-time control problem for which a coherent feedback protocol is optimal, and is faster than all measurement-based feedback protocols, where the latter are defined in a strict sense. The superiority of the coherent protocol is due to the fact that it can exploit a geodesic path in Hilbert space, a path that measurement-based protocols cannot follow. 
\end{abstract}

\maketitle 

\section{Introduction}

Feedback control, useful in reducing the effects of noise on quantum dynamical systems~\cite{Dong10, Jacobs14, Wiseman10, Altafini12}, can be realized in one of two ways: repeatedly measure the system and apply operations conditional on the measurement results~\cite{Jakeman85, Wiseman93b, DJ, Yamamoto07b,  Koch10, Vijay12, Brakhane12, Riste12}, or couple the system to another quantum system (a ``quantum controller'') thus avoiding the use of measurements~\cite{Wiseman94, Gough09, Nurdin09, Zhang12}. This second form of feedback is referred to as \textit{coherent feedback}. Our motivation here is that numerical optimization performed on linear feedback networks by Nurdin, James, and Petersen~\cite{Nurdin09} showed that coherent feedback can outperform feedback mediated by measurements, but a clear dynamical explanation for this intriguing difference is yet missing. Just recently it has been shown that coherent feedback can do very much better than measurement-based feedback when using a linear interaction to cool a linear oscillator~\cite{Hamerly12, Hamerly13}. It has been suggested that passing the noise signal through the auxiliary system coherently allows it to cancel the quantum noise in a way that a measured signal cannot~\cite{Hamerly12, Nurdin09}, but it is not obvious how to design protocols to exploit the coherent advantage. 

Here we present a clear dynamical explanation for the superior performance of coherent feedback, for linear and nonlinear systems. The task we consider is that of using an auxiliary system to extract all the entropy from a ``target'' system in the shortest time, given that the norm of the interaction Hamiltonian between the target system and the auxiliary is bounded. We show that a coherent protocol can extract all the entropy from the target system faster than any measurement-based protocol. The reason for this is that the coherent protocol can follow a more efficient path in the state-space. Since entropy extraction is at the heart of all feedback protocols that regulate a system against noise, this result has profound implications. We expect that the notion of a geodesic, the most efficient path in state-space, will be useful in designing superior continuous-time coherent feedback protocols. 

\section{The relationship between coherent and measurement-based feedback}

We must first elucidate the precise relationship between measurement-based and coherent control. In doing so it is conceptually useful to define measurement-based feedback in both a strict sense and a broad sense. To begin, coherent feedback control is a scenario in which a system to be controlled (the ``target system'', from now on just ``the system'') is coupled to another quantum system (``the auxiliary'') for the purpose of implementing control. This scenario is \textit{feedback} control so long as the auxiliary has no direct access to the noise driving the system; in this case all information exploited by the auxiliary must be obtained from the system, and so the process is one of feedback~\cite{Jacobs20Q, Jacobs14}. We do not need an explicit feedback ``loop'' in the definition of feedback for the same reason that there is no explicit loop in the Hamiltonian description of the Watt governor~\cite{Maxwell1868}. 

In a measurement-based feedback (MF) protocol a sequence of measurements is made on the system, and its Hamiltonian is modified after each measurement in a manner determined by the measurement result.  We will refer to a single measurement with its subsequent feedback step as a single iteration of an MF protocol. Every quantum measurement can be modeled as an interaction with a probe system, in which each measurement result corresponds to a subspace of the probe. The subspaces corresponding to the set of all the measurement results are mutually orthogonal. Because of this it is simple to realize feedback within this model: one merely choses a different Hamiltonian for the system for each of the orthogonal subspaces. The result is an implementation of MF using only a coherent interaction between the system and an auxiliary (in this case the probe), and as a result all MF protocols can be implemented using coherent feedback as defined above. Examples of the implementation of measurement-based feedback using unitary operations can be found in, e.g.~\cite{VItali04, Jacobs09c, Sarovar05}. 

While MF is subsumed by coherent feedback, the reverse is not true. The reason for this is that each iteration of an MF protocol is divided into two steps, and the second (the feedback step) is restricted to a certain form: an action performed on the system alone that is conditional on a set of mutually orthogonal subspaces of the auxiliary. If we denote the set of projectors onto the auxiliary subspaces by $\{\Pi_m\}$, then the unitary that acts during the feedback step must have the form 
\begin{equation}
   U_{\mss{fb}} = \sum_m U_m \otimes \Pi_m ,   \label{Ufb}
\end{equation}
where the $M$ unitaries $U_n$ act on the system. Further, just prior to the implementation of the feedback step the controller is decohered in the basis $\Pi_m$. This decoherence process is part of the definition of measurement, since the possible measurement outcomes are classical numbers that can be recorded by an observer, and thus cannot have quantum coherences between them.  The decoherence applies the operation $\rho \rightarrow \sum_m \Pi_m \rho \Pi_m$. We note that in fact the decoherence part of the measurement process is not required for measurement-based feedback, or for our analysis below.

In considering the measurement step of an MF iteration, it is important to note that the general formulation of an efficient quantum measurement implicitly includes feedback. An efficient measurement is described by a set of operators $\{A_m\}$, where $\sum_m A_m^\dagger A_m = I$, the probability of the $m^{\mss{\textit{th}}}$ outcome is $p_m = \mbox{Tr}[A_m^\dagger A_m\rho]$, and the state following the measurement is $\tilde{\rho}_m = A_m\rho A_m^\dagger/p_m$. Every operator $A_m$ can be factored via the polar decomposition theorem as $A_m = U_m P_m$, where $U_m$ is unitary and $P_m$ is positive. Since $U_m$ is unitary it is equivalent to an explicit feedback step applied after a measurement described by the set of operators $\{P_m\}$. We now have a choice in defining measurement-based feedback. Under the strictest definition, all unitary actions are restricted to the explicit feedback part, and as a result the measurement must be described by a set of positive operators. Measurements in this class we will refer to as being ``bare''~\cite{Jacobs14}, and we will refer to the class of MF protocols for which the measurements are bare as \textit{strict-sense} MF. Many real feedback protocols belong to this class. Strict-sense MF is a significantly smaller class of protocols than coherent feedback because the unitary operators are restricted in both steps. Our primary result is that this restriction reduces the speed at which certain joint operations can be performed. 

It is also important to consider how our result impacts measurement-based feedback control when defined in the broadest sense. That is, when we allow the measurement step to include any generalized measurement. While the measurement step can now perform \textit{any} coherent feedback process, the distinction between MF and coherent feedback (CF) remains, due to the fact that an MF protocol is only a feedback protocol if the feedback step contributes to (improves) the performance of the control. Since the geodesic path for our task cannot be broken into a measurement and feedback step, the feedback part of the protocol cannot follow this geodesic. Measurement-based feedback can therefore only follow the shortest path when the feedback step vanishes.    

\section{The fastest way to prepare a pure state} 
\label{geo}

\textit{\textbf{The task:}} to take an $N$-state system that is initially in an arbitrary mixed state to a specified pure state in the shortest time. This simple control task is a central element of a large number of control problems, because the faster a continuous-time protocol can reduce the entropy of a system the lower the resulting steady-state noise. 

\vspace{1ex}\noindent\textit{\textbf{The constraint on the speed of control:}} We impose the physically relevant constraint that the norm of the interaction Hamiltonian, $H$, between the system and auxiliary is bounded. Specifically, we bound the difference between the maximum and minimum eigenvalues of $H$ by $2\mu$. The fact we choose this particular norm does not further restrict the generality of the result, since all norms in a finite dimensional space are equivalent. The reason that we chose this norm is because the constant $\mu$ has a clear physical meaning. It is the maximal speed of rotation that the Hamiltonian can generate in Hilbert space. We state this as a theorem, which is proved in~\cite{Anandan90, Carlini06}:

\begin{theo}
\label{th::ensforrho}
\textnormal{\hspace{-1mm} Let $|0\rangle$ and $|1\rangle$ be two orthogonal quantum states. If the difference between the maximum and minimum eigenvalues of a Hamiltonian $H$ is $2\hbar\mu$, then 
\vspace{1ex}\newline i)  The fastest way that $H$ can transform $|0\rangle$ to a state $|\psi\rangle = \sqrt{\alpha} |0\rangle + \sqrt{1-\alpha}|1\rangle$ is via the path $|\phi(t)\rangle = \cos(\mu t) |0\rangle + \sin(\mu t)|1\rangle$. This path also has the shortest length between the two states, and is therefore a \textit{geodesic}. 
\vspace{1ex}\newline ii) The minimum time required for the transformation is $\tau = \theta/\mu$, where $\theta = \cos^{-1}(\alpha)$ is the angle between the initial and final state-vectors. 
\vspace{1ex}\newline iii) The term in $H$ that is required to generate the transformation $|0\rangle \rightarrow |1\rangle$ is
\begin{equation}
    H^{\mss{geo}} = \hbar \mu \left(  |0\rangle \langle 1|  +  |1\rangle \langle 0|  \right) . \label{Hgeo01}
\end{equation}
}
\end{theo}

\vspace{1ex} In addition to the above theorem, to obtain the optimal protocols below we will need an answer to the following question. If we wish to transform a single state, $|\psi\rangle$, to one of a set of $N$ orthogonal states, what is the minimum time required? More precisely, how small can we make the distance (the angle of rotation) to the state that is furthest from $|\psi\rangle$? The following lemma answers this question. 
    
\vspace{1mm}
\begin{lemma}
\label{lemm1}
\textnormal{\hspace{-1mm} Consider an arbitrary state $|\psi\rangle$ and a set of $N$ orthogonal states $|k\rangle$, $k=1,\ldots,N$. Denote the $N$ angles between $|\psi\rangle$ and each of the states $|k\rangle$ by $\theta_k = \cos^{-1}(|\langle\psi |k\rangle|)$, and the largest of these angles by $\theta_{\mss{max}}$. If we are allowed complete freedom in choosing the set of orthogonal states $\{|k\rangle\}$, then the smallest value of $\theta_{\mss{max}}$ we can achieve is 
\begin{equation}
   \min_{\{|k\rangle\}} \theta_{\mss{max}} = \cos^{-1}(1/\sqrt{N}), 
\end{equation}
and this is only achieved when $|\psi\rangle$ is unbiased with respect to the states $\{|k\rangle\}$. That is, when $\theta_k = \cos^{-1}(1/\sqrt{N})$, $\forall k$. 
}
\end{lemma}

\begin{proof}
Since the states $\{|k\rangle\}$ are orthogonal, $|\psi\rangle = \sum_{k=0}^{N-1} c_k |k\rangle + z |\chi\rangle$ where $|\chi\rangle$ is orthogonal to all the $|k\rangle$. We have also $\sum_{k=0}^{N-1} |c_k|^2 = 1 - |z|^2$. Now let us take $|\psi\rangle$ to be unbiased with respect to the set $\{|k\rangle\}$. In this case $|c_k|^2 = (1 - |z|^2)/N$ for every $k$. Now, since $\sum_{k=0}^{N-1} |c_k|^2 = N/K$ for every orthogonal set $\{|k\rangle\}$, if we increase the moduli of any of the $c_k$ we must reduce at least one other of these moduli, thus increasing one of the angles $\theta_k$ above the value $\cos^{-1}(1/\sqrt{K})$. The value of  $\theta_{\mss{max}}$ can therefore be no less than $\cos^{-1}(1/\sqrt{K})$. 
\end{proof}

Combining the above theorem and lemma, the minimum time required to transform a single state to any one of $N$ mutually orthogonal states, under the Hamiltonian constraint above, is $t_{\mss{min}} = \cos^{-1}(1/\sqrt{N})/\mu$.  

\subsection{The fastest measurement-based protocol}

How fast can strict-sense MF perform our task? Well, for the bare measurement to prepare the system in a pure state each of the measurement operators must be a rank-one projector. We can therefore write the measurement operators as $A_k \propto |\chi_k\rangle \langle\chi_k |$ for some set of $K$ states $|\chi_k\rangle$ that span the $N$-dimensional space of the system. To realize this measurement the interaction with the auxiliary system, followed by decoherence in an auxiliary basis that we will denote by $\{ |k\rangle \}$, must prepare the joint system in the form 
\begin{equation} 
   \tilde{\rho} = \sum_{k=1}^{K-1} q_k |\chi_k\rangle \langle \chi_k | \otimes  |k\rangle \langle k| .
\end{equation}
Since the auxiliary starts in a single pure state, by lemma 1 the minimum angle required to rotate this state to each of the $K$ basis states $|k\rangle$ is $\cos^{-1}(1/\sqrt{K})$. The minimum value of $K$ is $N$, and so the minimum angle is $\cos^{-1}(1/\sqrt{N})$. This gives us the minimum time for the measurement step.  

To leave the system in a single pure state, the feedback step must rotate each of the states $|\chi_k\rangle$ to this state. If the set $\{|\chi_k\rangle\}$ were a basis, then our lemma above would tell us that the minimal angle required to rotate all of them to the final pure state would be $\cos^{-1}(1/\sqrt{N})$. But the $|\chi_k\rangle$ need not be orthogonal. Not surprisingly the minimum angle is still $\cos^{-1}(1/\sqrt{N})$ for an arbitrary set of states that decomposes the $N$-dimensional identity operator, but it is a little more involved to prove. We now extend lemma 1 to this case. 

\vspace{1ex}\noindent\textit{\textbf{Extending lemma 1:}} Let us assume that each of the $K$ states $|\chi_k\rangle$ has an overlap, $o_k$, with the state $|\psi_0\rangle$, that is greater than or equal to $1/\sqrt{N}$, and that the density matrix $w = \sum_{k=1}^{K-1} q_k |\chi_k\rangle \langle \chi_k |$ is maximally mixed. Since the probability for being in the state $|\psi_0\rangle$ is $ \langle \psi_0 |w|\psi_0\rangle = 1/N$, we must have $\sum_{k=1}^K |o_k|^2 q_k = 1/N$. Since $|o_k|^2 \geq 1/N$ by assumption, we must have $\sum_{k} q_k \leq 1$, with equality if and only if $|o_k|^2 = 1/N$ for all $k$. Now consider the probability that the state of the system is not $|\psi_0\rangle$, which is $Z = 1 - 1/N$, and is given by $\sum_{k=1}^K (1 - |o_k|^2) q_k$. Since $1 - |o_k|^2 \leq 1 - 1/N$, $Z$ can only be equal to $1 - 1/N$ if $\sum_{k} q_k = 1$, and from the analysis above, if none of the $|o_k|^2$ are less than $1/N$, this can only be true if all the $|o_k|^2$ are equal to $1/N$. We can therefore conclude that at least one of the overlaps $o_k$ must be no greater than $1/\sqrt{N}$.  

\vspace{1ex}

Since our protocol must prepare the system in a desired pure state, $|\psi_0\rangle$, for every initial state, in determining the requirements for the protocol, we have the freedom to chose this initial state. We will assume that the initial state is maximally mixed, as this will simplify our analysis. A theorem by Ando states that for any bare measurement, the density matrix 
\begin{equation}
   w = \sum_{k=1}^{K-1} q_k |\chi_k\rangle \langle \chi_k | 
\end{equation}
satisfies $S(w) \geq S(\rho)$, where $S$ is the von Neumann entropy~\cite{Ando89, Jacobs14}. Since $\rho = I/N$, where $I$ is the identity, we have also $w = I/N$. We can therefore conclude from the extension of lemma 1 above that the minimum angle required to rotate each of the $|\chi_k\rangle$ to $|\psi_0\rangle$ is the same as that required by the measurement step, being $\cos^{-1}(1/\sqrt{N})$. The fastest measurement-based feedback protocol therefore take a time of 
\begin{equation}
    \tau_{\mss{mf}} = \frac{2}{\mu} \cos^{-1}\left(\frac{1}{\sqrt{N}}\right) .   \label{tmeas}
\end{equation}

\subsection{The optimal protocol} 

We first consider a straightforward state-swapping protocol, which we will find is not optimal, but is faster than all measurement-based protocols for $N \geq 3$. We can prepare the system in a pure state by bringing up an auxiliary that is already in a pure state, and swapping the states of the two systems. If we denote a basis for the joint system by $\{| n\rangle |k\rangle\}$, meaning the tensor product of system state $|n\rangle$ with auxiliary state $|k\rangle$, then a swap operation simply swaps the indices of the system and auxiliary states:
\begin{equation}
       |n\rangle |k\rangle  \rightarrow  | k\rangle |n\rangle, \;\; \forall \; n,k  \;\;\;\; \mbox{(state swap)} . 
\end{equation}
We can perform a swap by choosing a joint Hamiltonian to perform the above transformation for each joint state $|n\rangle |k\rangle$. Since the initial and final states are orthogonal except when $n=k$, the Hamiltonian must rotate the states through $\theta=\pi/2$, and thus requires at time of 
\begin{equation}
        \tau_{\mss{swap}} = \frac{\pi}{2\mu} . 
\end{equation}

We now show that there is a faster protocol that will also swap the states of the two systems. We first note that we can perform local unitary transformations on both systems, after the swap has been performed, and this will not affect the entropy transferred by the swap. We can therefore implement the entropy transfer by using the transformation 
\begin{equation}
       |n\rangle |k\rangle  \rightarrow  | \psi_k \rangle |\phi_n\rangle, \;\; \forall \; n,k  \;\;\;\; \mbox{(state swap)} . 
\end{equation}
where $\{| \psi_k \rangle \}$ and $\{| \phi_n \rangle \}$ are any bases of the system and auxiliary, respectively. We now examine what happens when we chose the basis $\{| \psi_k \rangle \}$ to be unbiased with respect to $\{| n \rangle \}$ and $\{| \phi_n \rangle \}$ to be unbiased with respect to $\{| k \rangle \}$, so that 
\begin{equation}
      \langle \psi_k |n\rangle =  \langle \phi_n |k\rangle = \frac{1}{\sqrt{N}} . 
\end{equation}
For this new swap operation, for each $n$ and $k$ the final state $| \psi_k \rangle |\phi_n\rangle$ is no longer orthogonal to the initial state $|n\rangle |k\rangle$. The overlap between the initial and final states is the same for every $n$ and $k$, being 
\begin{equation}
   ( \langle \psi_k | \langle \phi_n |) |n\rangle |k\rangle  = \frac{1}{N} . 
\end{equation}
The time required for this swap operation is therefore  
\begin{equation}
        \tau_{\mss{opt}} = \frac{1}{\mu}\cos^{-1}\left( \frac{1}{N}\right) .   \label{topt}
\end{equation}

We now show that our second swap protocol is optimal. We first recall that from the proof of lemma 1 we know that if we have $N$ positive numbers $r_n$ that sum to a fixed value, then the way to make he smallest of these numbers maximal is to make them all equal. We now observer that for our protocol to transfer all of the entropy from the system, it must map a full set of basis states, $\{|n\rangle\}$ to a single state, $|\psi\rangle$. To satisfy unitarity, this means that the initial state of the auxiliary, which we will denote by $|\phi\rangle$, must be rotated to $N$ orthogonal states. Denoting these $N$ states of the auxiliary by $\{|\phi_n\rangle\}$. The transformation must therefore perform the map 
\begin{equation}
       |n\rangle |\phi\rangle  \rightarrow  | \psi \rangle |\phi_n\rangle, \;\; \forall \; n  .   \label{minter}
\end{equation}
To perform this transformation in the shortest time, we need to chose the bases $\{|n\rangle\}$ and $\{|\phi_n\rangle\}$ to maximize the minimum overlap between the initial and final states. Let us denote these $N$ overlaps by $r_n =   ( \langle \psi | \langle \phi |) |n\rangle |\phi_n\rangle$. If we write $| \psi \rangle$ in the basis $| n \rangle$, and $|\phi\rangle$ in the basis $\{|\phi_n\rangle\}$ as 
\begin{equation}
     | \psi \rangle = \sum_n  u_n |n\rangle , \;\;\;\;\;  |\phi\rangle  =   \sum_n  v_n |\phi_n\rangle ,  
\end{equation}
then
\begin{equation}
    Z \equiv  \sum_n r_n  = \mathbf{u} \cdot \mathbf{v}  \leq  \sqrt{(\mathbf{u} \cdot \mathbf{u})(\mathbf{v} \cdot \mathbf{v})} = 1,  
\end{equation}
where $\mathbf{u} = (u_1,\ldots, u_N)$, $\mathbf{v} = (v_1,\ldots, v_N)$, and the inequality is the Cauchy-Schwartz inequality. To maximize the smallest $r_n$ we need to maximize $Z$, which means choosing $\mathbf{u} = \mathbf{v}$ so that $r_n = |u_n|^2$. We thus maximize the minimum $r_n$ by setting $u_n = v_n = 1/\sqrt{N}, \forall n$. The minimum rotation angle required to perform the transformation in Eq.(\ref{minter}) is therefore $\theta_{\mss{min}} = \cos^{-1}(1/N)$, and this is achieved by our second protocol above.

\section{Discussion}
We have considered the time required to prepare a system in a pure state under a constraint on the speed of evolution in Hilbert space. We have shown that the fastest measurement-based protocol, defined in the strict sense, takes a time given by Eq.(\ref{tmeas}), and the optimal protocol takes a shorter time given by Eq.(\ref{topt}). For an $N$-dimensional system the optimal protocol is faster by a factor of 
\begin{equation}
        s(N) = \frac{2 \cos^{-1}\left(1/\sqrt{N}\right)}{\cos^{-1}\left(1/N\right)} . 
\end{equation}
For a single qubit this factor is $s(2) = 1.5$. It increases monotonically with $N$, and for large $N$ we have 
\begin{equation}
        s(N) = \frac{2 \cos^{-1}\left(1/\sqrt{N}\right)}{\cos^{-1}\left(1/N\right)} \approx 2 - \frac{1}{\sqrt{N}} , \;\;\;\; N \gg 1 . 
\end{equation} 
While we expect our results to help elucidate situations in which coherent feedback is a better choice than measurement-based feedback, we also expect that they will help in the design of coherent feedback protocols. It may be useful to consider the ideal transformation that will perform the desired task, and then consider implementing this transformation using geodesics. Firstly our analysis suggests that when designing coherent feedback protocols, translating measurement-based protocols directly into coherent protocols is not the best approach. As an example, the design of coherent protocols that implement quantum error correction (QEC) is a natural candidate for such an approach~\cite{Sarovar05}. The process of correcting errors in quantum encoded data is a single-shot feedback process, in which the ``error-syndrome'' is measured and the information used to apply the correction~\cite{Gottesman09}. Rather than translating this process directly into coherent measurement and feedback steps, geodesics could be used to design a faster implementation. 

Secondly, we have shown that the use of a Hamiltonian that implements the usual swap operation is not necessarily the fastest way to transfer entropy. Superior coherent protocols might be obtained by generalizing the process by which we found the optimal purification protocol. That is, by finding the class of transformations that will perform the desired function, and then searching for those that have the shortest geodesics. 

\section*{Acknowledgments} 

KJ and XW are partially supported by the NSF projects PHY-1005571 and PHY-1212413, and by the ARO MURI grant W911NF-11-1-0268. HMW is partially supported by the ARC Centre of Excellence CE110001027. 

\vspace{5mm}


\end{document}